\documentclass[english]{elsarticle}
\usepackage[T2A,T1]{fontenc}
\usepackage[latin9]{inputenc}
\usepackage{float}
\usepackage{amsthm}
\usepackage{amsmath}
\usepackage{amssymb}
\usepackage{graphicx}
\usepackage{esint}

\makeatletter

\newcommand{\noun}[1]{\textsc{#1}}
\DeclareRobustCommand{\cyrtext}{%
  \fontencoding{T2A}\selectfont\def\encodingdefault{T2A}}
\DeclareRobustCommand{\textcyr}[1]{\leavevmode{\cyrtext #1}}
\AtBeginDocument{\DeclareFontEncoding{T2A}{}{}}

\newcommand{\lyxmathsym}[1]{\ifmmode\begingroup\def\b@ld{bold}
  \text{\ifx\math@version\b@ld\bfseries\fi#1}\endgroup\else#1\fi}

\numberwithin{equation}{section}
\numberwithin{figure}{section}
\theoremstyle{plain}
\newtheorem{thm}{\protect\theoremname}
  \theoremstyle{definition}
  \newtheorem{defn}[thm]{\protect\definitionname}
  \theoremstyle{plain}
  \newtheorem*{lem*}{\protect\lemmaname}
 \theoremstyle{definition}
 \newtheorem*{defn*}{\protect\definitionname}
  \theoremstyle{plain}
  \newtheorem*{thm*}{\protect\theoremname}

\makeatother

\usepackage{babel}
  \providecommand{\definitionname}{Definition}
  \providecommand{\lemmaname}{Lemma}
  \providecommand{\theoremname}{Theorem}
\providecommand{\theoremname}{Theorem}

\begin{document}

\title{\textbf{Solid Continuum with Thermofluctuation Kinetics of Microcracks.
Phase Transition.}}

\author{{\normalsize{V. M. Gertsik. }}%
\thanks{Institute of Earthquake Prediction Theory and Mathematical Geophysics
RAS, Moscow, RF, getrzik@ya.ru%
}{\normalsize{, A.L.Petrosyan }}%
\thanks{Institute of Earthquake Prediction Theory and Mathematical Geophysics
RAS, Moscow, RF, alikmitpan@yandex.ru%
}}
\begin{abstract}
{\normalsize{Thermodynamic equations for a solid and a solid continuum
under stress are derived on the basis of a multicomponent mean field
Markov process for thermofluctuation kinetics of microcracks. The
resulting continuum is viscous elastoplastic continuum with damage.
It can radiate elastic waves . The existence of phase transitions
with microcrack density as an order parameter is proved for a stationary
state of a special model of solid.}} {\normalsize{For a finite large
system the distribution of the logarithmic power of acoustic emission
at a critical point is similar to the distribution of the logarithmic
energy of earthquakes.}}{\normalsize \par}
\end{abstract}
\maketitle

\section*{INTRODUCTION}

This paper presents a system of equations for the non-equilibrium
dynamics of an elastic continuum involving the appearance and healing
of microcracks. The basis for deriving the equations is a multicomponent
continuous-time mean field Markov process with intensities of the
activation type. Since the contribution of a crack to the deformation
of the body is memorized after the healing, the resulting behavior
is nonlinear viscous elastoplasticity. This type of behavior demonstrates
qualitative correspondence with the real properties of a solid. For
instance, a solid shows brittleness and elasticity at low temperatures
and/or large rates of deformation and fluidity and plasticity at high
temperatures and/or small rates of deformation. 

For simplest model the existence of phase transition between phases
of a low and high density of microcracks is strictly proved. An area
of the high density phase is a soft inclusion in the continuum as
its pliability is increased. For certain types of stress field this
area can propagate in the continuum like a crack and radiate elastic
waves. Therefore this continuum can serve as model for earthquake
generation. 

This system of equations expands the range of known equations for
continuous continuum such as the gas dynamics equations, Navier\textendash{}Stokes
equations etc.

The idea of using the activation principle in crack kinetics dates
back to the kinetic concept of strength suggested by S.N.Zhurkov \citet{Zhurkov-1}
The empirical Zhurkov formula

\begin{equation}
\tau=\tau_{0}\;\exp\left\{ \frac{U-\nu\sigma}{k_{B}T}\right\} \label{eq:0-1}
\end{equation}
 describes lifetime $\tau$ of a specimen under tensile load $\sigma$
at temperature $T$, where $k_{B}$ is Boltzman's constant, $\tau_{0,}\; U$
and $\mbox{\ensuremath{\nu}}$ are constants. In spite of the fact
that Zhurkov's formula clearly indicates the thermofluctuation mechanism
of fracture, numerous attempts to create on this basis a mathematical
apparatus of the theory have not been successful. The reason seems
to lie in the fact that the quantity $1/\tau$ cannot be directly
used as intensity of microfracture generation, Firstly, $\tau$ is
a macroscopic quantity and cannot serve as ``a first principle''
but must be calculated from equations known beforehand. Secondly,
the energy in the numerator of the expression in the braces must be
a quadratic, but not a linear, function of $\sigma$, because the
linearity contradicts the definition of elastic energy.

However, if we use an expression of the type $1/\tau$ with quadratic
dependence on the stress as the intensity of microcrack appearance
, physically correct equations can then be derived. Moreover, as will
be demonstrated below, experimental relations of specimen lifetime
$\tau$ versus $\sigma$ and $T$ presented by Zhurkov as (\ref{eq:0-1})
are reproduced in our model..

\section{MULTICOMPONENT MEAN FIELD PROCESS}
\begin{defn}
Let $\bar{\xi}_{N}(t)\equiv\{\xi_{N}(x,t),$ $x\in\Omega_{N}\},\,\xi_{N}(x,t)=0,1,\,|\Omega_{N}|=N,\, t\text{\ensuremath{\ge}}0,$
($|A|$ denotes the number of elements in $A$) be a $N$-component
continuous-time Markov process with state-space $\left\{ 0,...,K\right\} ^{\Omega_{N}}$.
Denote $\mathrm{\boldsymbol{y}}(t)=\left\{ \frac{n_{1}(t)}{N},...,\frac{n_{K}(t)}{N}\right\} ,$$n_{k}(t)\equiv n_{k}(\bar{\xi}_{N}(t))=|\{x\in\Omega_{N}:\,\xi_{N}(x,t)=k\}|,\: k=1,...,K.$
and let $\Lambda_{k}(\mathrm{\boldsymbol{z}}),\mathrm{\, M_{\mathit{k}}(\mathit{\boldsymbol{\mathrm{z}}})},\:\mathrm{\mathbf{z}}=(z_{1},...,z_{K})$$,$
be positive continuous functions on $[0,1]^{K}$ . We assume that
only point translations $\left\{ 0\rightarrow k,k\rightarrow0\right\} ,k=1,...,K,$
are possible with the rates given by conditional probabilities 
\begin{gather}
\begin{cases}
\Pr\{\xi_{N}(x,t+h)=k,\, k>0|\xi_{N}(x,t)=0,\mathrm{\boldsymbol{y}}(t)\}=\Lambda_{k}(\mathrm{\boldsymbol{y}}(t))h+o(h),\\
\Pr\{\xi_{N}(x,t+h)=0|\xi_{N}(x,t)=k,\, k>0,\mathrm{\boldsymbol{y}}(t)\}=\mathrm{M}_{k}(\mathrm{\boldsymbol{y}}(t))h+o(h)
\end{cases}\label{eq:1-1-1}
\end{gather}

\end{defn}
(As usual $\Pr\{A\}$ denotes the probability of $A$ and $\Pr\{A|B\}$
denotes the conditional probability of $A$ given $B$.) This process
will be called a ``\emph{mean field}'' Markov process.

.

Below we are going to prove that components $\xi_{N}(x,t)$ of stationary
Markov process $\bar{\xi}_{N}$ with $K=1$ are asymptotically independent
at $N\rightarrow\infty$(see Appendix).. 

For stationary processes with $K\geq2$ and non-stationary processes
$\bar{\xi}(t)\equiv\bar{\xi}_{\infty}(t)$ the proof of the independence
of components is unavailable. Therefore we need the additional definition
in these cases..
\begin{defn}
A mean field Markov process $\bar{\xi}\equiv\left\{ \xi(x,t),\, x\in\Omega,\, t\geq0\right\} $on
a countable set of points $\Omega$ is a set of independent Markov
processes $\xi(x,t),\, x\in\Omega,$ with state-space $\left\{ 0,..,,K\right\} $,
with conditional probabilities
\begin{gather*}
\begin{cases}
\Pr\{\xi_{N}(x,t+h)=k,\, k>0|\xi_{N}(x,t)=0,\mathrm{\boldsymbol{y}}(t)\}=\Lambda_{k}(\mathrm{\boldsymbol{p}}(t))h+o(h),\\
\Pr\{\xi_{N}(x,t+h)=0|\xi_{N}(x,t)=k,\, k>0,\mathrm{\boldsymbol{y}}(t)\}=\mathrm{M}_{k}(\mathrm{\boldsymbol{p}}(t))h+o(h)
\end{cases}
\end{gather*}
where $\boldsymbol{\mathrm{p}}\mathrm{(}t\mathrm{)}=\left\{ p_{k}(t),\, k=1,..,,K\right\} $,
$p_{k}(t)\mathrm{\equiv Pr}\left\{ \xi(x,t)=k\right\} $.\end{defn}
\begin{lem*}
The probabilities $p_{k}(t)$ are discribed by equations

\[
\frac{dp_{k}(t)}{dt}=\left[1-\sum_{i=1}^{K}p_{l}(t)\right]\mathrm{\Lambda_{\mathit{k}}}\left(\boldsymbol{\mathrm{p}}\mathrm{(}t\mathrm{)}\right)-p_{k}(t)\mathrm{M_{\mathit{k}}}\left(\boldsymbol{\mathrm{p}}\mathrm{(}t\mathrm{)}\right),\: k=1,...K.
\]
\end{lem*}
\begin{proof}
According to the law of large numbers $\underset{N\rightarrow\infty}{\lim}\frac{n_{k}(t)}{N}=p_{k}(t),\, k=1,...,K$
and the equations are easily deduced from \ref{eq:1-1-1} by the following
manipulation:
\begin{gather*}
p_{k}(t+h)=\underset{k^{\prime}=0,...K}{\sum}\Pr\{\xi_{N}(x,t+h)=k|\xi_{N}(x,t)=k^{\prime},\boldsymbol{\mathrm{p}}\mathrm{(}t\mathrm{)}\}\mathrm{Pr}\left\{ \xi(x,t)=k^{\prime}\right\} =\\
=\Pr\{\xi_{N}(x,t+h)=k,\, k>0|\xi_{N}(x,t)=k,\boldsymbol{\mathrm{p}}\mathrm{(}t\mathrm{)}\}p_{k}(t)+\\
+\Pr\{\xi_{N}(x,t+h)=k,\, k>0|\xi_{N}(x,t)=0,\boldsymbol{\mathrm{p}}\mathrm{(}t\mathrm{)}\}\left[1-\sum_{i=1}^{K}p_{l}(t)\right]=\\
=\left[1-\mathrm{M_{\mathit{k}}}\left(\boldsymbol{\mathrm{p}}\mathrm{(}t\mathrm{)}\right)h\right]p_{k}(t)+\mathrm{\Lambda_{\mathit{k}}}\left(\boldsymbol{\mathrm{p}}\mathrm{(}t\mathrm{)}\right)\left[1-\sum_{i=1}^{K}p_{l}(t)\right]h+o(h)=\\
=p_{k}(t)+\left\{ \left[1-\sum_{i=1}^{K}p_{l}(t)\right]\mathrm{\Lambda_{\mathit{k}}}\left(\boldsymbol{\mathrm{p}}\mathrm{(}t\mathrm{)}\right)-p_{k}(t)\mathrm{M_{\mathit{k}}}\left(\boldsymbol{\mathrm{p}}\mathrm{(}t\mathrm{)}\right)\right\} h+o(h).
\end{gather*}
The result follow from 
\[
\frac{dp_{k}(t)}{dt}=\underset{h\rightarrow0}{\lim}\frac{p_{k}(t+h)-p_{k}(t)}{h}.
\]

\end{proof}
.

\section{APPEARANCE AND HEALING OF MICROCRACKS}

We imagine an infinite solid volume to be subdivided into equal cubic
cells with edges perpendicular to coordinate axes $x_{1}$, $x_{2}$,
$x_{3}$ and with centers at the points of a cubic lattice $\mathbb{Z_{\mathrm{a}}^{\mathrm{3}}}$,
$a$ is the step of the lattice. We prescribe at boundaries of the
body a uniform symmetric tensor of second order (stress tensor) $\sigma\equiv\sigma_{kl}(t)$,
$k,l=1,2,3,$ that is a function of, in general, time $t$. A configuration
$\zeta_{t}$ has at $x\in\mathbb{Z_{\mathrm{a}}^{\mathrm{3}}}$ the
value $\zeta_{x,t}=u_{k},\, k=1,...,K.$ if the cell centered at $x$
contains a planar disk (microcrack) centered at the same point, with
radius $r<a/2$ and with normal $u_{k}$. The value $\zeta_{x,t}=u_{0}$
corresponds to an empty cell.

We define a multicomponent mean field birth-death process of microcracks
appearance and disappearance by equations

\begin{equation}
\frac{dp_{k}(t)}{dt}=\left[1-\sum_{i=1}^{K}p_{l}(t)\right]\mathrm{\Lambda_{\mathit{k}}}-p_{k}(t)\mathrm{M_{\mathit{k}}}\label{eq:9}
\end{equation}
with the intensities as defined below. 

The intensity $\mathrm{M}_{\mathit{k}}\equiv\mathrm{M_{\mathit{k}}}\left(\boldsymbol{\mathrm{p}}\mathrm{(}t\mathrm{)}\right)$
of microcrack healing has the form

\begin{equation}
\mathrm{M}_{\mathit{k}}=c_{0k}\;\exp\left\{ -\beta U_{k}(\sigma;\,\boldsymbol{\mathrm{p}}\mathrm{(}t\mathrm{)})\right\} \label{eq:2}
\end{equation}
where $c_{0k}$ are the numeric constants, $\beta=\frac{1}{k_{B}T},$
$T$ is temperature, $k_{B}$ is Boltzmann's constant, $\sigma$ is
the stress tensor, $U_{k}(\sigma;\,\boldsymbol{\mathrm{p}}\mathrm{(}t\mathrm{)})$
are activation energies of the healing. We use here the random truncation
of a continuous healing process. If the component of the stress tensor
normal to the microcrack plane is compressive, the crack is closed.
During the healing of the closed crack its opposite sides stick together
and their relative displacement is remembered and carries this contribution
to the residual strain. If the component of the stress tensor normal
to the microcrack plane is tensile, the crack is open. Its healing
is the filling of the cavity with molecules from the host material.
The strain brought about by the relative displacement of crack sides
is remembered but the volume part of the strain disappears, because
the density of the material remains unchanged.

.We suppose that the rate of healing is proportional to the rate of
diffusion. The diffusion coefficient $D$ is specified by the Hevesy
formula (see, e.g., \citet{Frenkel })

\[
D=D_{0}\exp\left\{ -\beta E_{a}\right\} ,
\]
where $D_{0}$ is constant and $E_{a}$ is activation energy. In our
case activation energies $U_{k}(\sigma;\,\boldsymbol{p}\mathrm{(}t\mathrm{)})$
depend in general on the stress normal to the microcrack plane and
are different for open and closed cracks. It is convenient to assume
that a crack exists without change until the moment of healing, at
which it disappears together with its contribution to the stress field.

To define the intensity $\mathrm{\mathrm{\Lambda}_{\mathit{k}}\equiv\Lambda_{\mathit{k}}}\left(\boldsymbol{p}\mathrm{(}t\mathrm{)}\right)$
of microcrack birth let's assume that any cell contains $K$ types
of microdefects with molecular size. A defect of each type can lose
its stability and become a microcrack of the same type.

The intensity $\mathrm{\lambda}_{\mathit{k}}$ has the form

\begin{equation}
\mathrm{\Lambda_{\mathit{k}}}=c_{1k}\;\exp\left\{ -\beta\left[H-E_{k}^{0}(\sigma;\,\boldsymbol{p}\mathrm{(}t\mathrm{)})\right]\right\} \label{eq:3}
\end{equation}
where $c_{1k}$ are constants, $H$ is activation energy (that is,
the thermofluctuative elastic energy threshold where a microdefect
loses stability and becomes a microcrack), $E_{k}^{0}(\sigma;\,\boldsymbol{p}\mathrm{(}t\mathrm{)})$
is the additional elastic energy brought in by a $k$-type microdefect.

\section{MEAN STRESS FIELD }

The stress field in a body with an arbitrary microcrack system cannot
be represented explicitly. Therefore we use an approximation; roughly
speaking, we assume that a reduction of stress in the neighborhoods
of cracks is compensated by an increase of stress outside of these
neighborhoods, the mean stress over the volume being kept equal to
its value at the boundary. (A similar approach is used in problems
arising in breaking of ropes composed of many wires and also in strength
models for solid bodies under axial tension.)

We assume that the stress tensor outside of spheres of radii $r$
circumscribed around microcracks is uniform and is specified by an
effective tensor $\overline{\sigma}$ . Inside of these spheres the
stress is also uniform and differs from $\overline{\sigma}$ by some
zero components when they vanish at the crack surface . It is clear
that we only approximate a continuous stress field by discontinuous
functions and by no means assume that stress undergoes actual discontinuities
at the surface of the spheres.

Let us choose coordinates $x_{1}^{i}$, $x_{2}^{i}$, and $x_{3}^{i}$
for $i$-type microcracks in such a way that the $x_{1}^{i}$ axis
is perpendicular to microcrack planes. Denote by $\mathbf{A}_{kl}^{(i)}$
the matrix elements specifying the transformation $\mathbf{A}^{(i)}$
of coordinates $x_{1}$, $x_{2}$, $x_{3}$ to $x_{1}^{i}$, $x_{2}^{i}$,
$x_{3}^{i}$. Let $\boldsymbol{\mathbf{\mathbf{T}}}^{(i)}\overline{\sigma}$
: $(\boldsymbol{\mathbf{\mathbf{T}}}^{(i)}\overline{\sigma})_{kl}=\mathbf{A}_{km}^{(i)}\mathbf{A}_{ln}^{(i)}\overline{\sigma}_{mn}$
be the tensor $\overline{\sigma}$ in the new coordinates (a notational
convenience introduced by Einstein will be used hear: tensor sums
are taken over all repeated subscripts). We introduce the piecewise
linear operator $\mathtt{\boldsymbol{\mathbb{\mathbb{\mathbf{S}}}}}$
by the rule: $(Sa)_{kl}=S_{kl}a_{kl}$ (in this unique case the summation
over repeated subscripts isn't made) for any symmetric tensor $a$,
where $\mathtt{\boldsymbol{\mathbb{\mathbb{\mathbf{S}}}}}=\left|\begin{array}{ccc}
0 & 0 & 0\\
0 & 1 & 1\\
0 & 1 & 1
\end{array}\right|$ if $a_{11}\geq0$, $\mathtt{\boldsymbol{\mathbb{\mathbb{\mathbf{S}}}}}=\left|\begin{array}{ccc}
1 & 0 & 0\\
0 & 1 & 1\\
0 & 1 & 1
\end{array}\right|$ if $a_{11}<0$. The case $a_{11}\geq0$ corresponds to the tensile
normal stress applied on the plane of an open crack. The case $a_{11}<0$
corresponds to the compressive normal stress applied on the plane
of a closed crack. Let us define an operator $\boldsymbol{\mathbf{\mathbf{Q}}}^{(i)}$
which is applied to $\overline{\sigma}$ by

.
\begin{eqnarray*}
\boldsymbol{\mathbf{\mathbf{Q}}}^{(i)}\overline{\sigma}=\boldsymbol{\mathbf{\mathbf{T}}}^{(i)-1}\mathbf{S}\boldsymbol{\mathbf{\mathbf{T}}}^{(i)}\overline{\sigma}.
\end{eqnarray*}

Piecewise linear operators $\boldsymbol{\mathbf{\mathbf{Q}}}^{(l)}$
specifies stress in the neighborhood of cracks and remove those stress
tensor components that are tangent to the crack plane and normal to
the crack plane component if it is positive, i.e. tensile (it is the
case of an ``open'' crack).

We assume that $\theta=\frac{4\pi r^{3}}{3a^{3}},$ then $\theta p_{i}(t)$
is the relative volume with stress $\boldsymbol{\mathbf{\mathbf{Q}}}^{(i)}\overline{\sigma}$,
and the tensor $\overline{\sigma}$ is the solution of the system
of equations

\noun{
\begin{gather}
\left[\mathbf{I}-\sum_{i=1}^{K}\theta p_{i}(t)\left(\mathbf{I}-\boldsymbol{\mathbf{\mathbf{Q}}}^{(i)}\right)\right]\overline{\sigma}=\sigma.\label{eq:4}
\end{gather}
}

\section{ELASTIC ENERGY}

The density of elastic energy $e(\sigma)$ in a homogeneous body under
stress $\sigma\equiv\sigma_{ij}$ is given by 
\[
e(s)=\frac{(\sigma\mathbf{:}\varepsilon)}{2}=\frac{(\sigma\mathbf{:}\mu\sigma)}{2}=\frac{(\lambda\varepsilon\mathbf{:}\varepsilon)}{2}
\]
(\citet{Landau-1}), where $\varepsilon\equiv\varepsilon_{kl}$ is
the elastic strain tensor, the tensor of the fourth order $\lambda\equiv\lambda_{ijmn}$
is the stiffness tensor, $\mu\equiv\mu_{ijmn}$ is inverse for $\lambda$
elastic pliability tensor, $\sigma=\lambda\varepsilon$ ($\sigma_{ij}=\lambda_{ijmn}\varepsilon_{mn}$)
is generalized Hooke's law, $(\sigma\mathbf{:}\varepsilon)\equiv\sigma_{ij}\varepsilon_{ij}$
is double inner product of tensors. (We use in this paper the approximation
of small deformations, where Hooke's law is fulfilled at all values
of the strain tensor.)

In a homogeneous isotropic body under the stress $\sigma$ the density
of elastic energy $\tilde{e}(\sigma)$ has the form

\begin{gather*}
\tilde{e}(\sigma)=\frac{1}{2E}(\sigma_{11}^{2}+\sigma_{22}^{2}+\sigma_{33}^{2})+\frac{1+\nu}{E}(\sigma_{12}^{2}+\sigma{}_{23}^{2}+\sigma_{31}^{2})
\end{gather*}
where $E$ is Young's modulus and $\nu$ is Poisson's ratio. Let $A$
be a body constructed from a cubic set of cells in $\mathbb{Z_{\mathrm{a}}^{\mathrm{3}}}$,
$n_{k}$ be the number of cells with $k$-type cracks, $\boldsymbol{p}^{(A)}=(p_{1}^{(A)},\ldots,p_{K}^{(A)}),$
$p_{k}^{(A)}=\frac{n_{k}}{|A|},\:\rho_{k}^{(A)}=\theta p_{k}^{(A)}$
. It is easy to see that the elastic energy $E_{A}(\sigma;\,\boldsymbol{p}^{(A)})$
of the body with cracks under the stress $\sigma$ is given by

\begin{gather*}
E_{A}(\sigma;\,\boldsymbol{p}^{(A)})=a^{3}|A|e_{A}(\sigma;\,\boldsymbol{p}^{(A)})=\\
=a^{3}|A|\left\{ \left[1-\sum_{k=1}^{K}\rho_{k}^{(A)}\right]\tilde{e}(\overline{\sigma})+\sum_{k=1}^{K}\rho_{k}^{(A)}\tilde{e}(\mathbf{Q}^{(k)}\overline{\sigma})\right\} =\\
=a^{3}|A|\left\{ \left[1-\sum_{k=1}^{K}\rho_{k}^{(A)}\right]\overline{\sigma}_{ij}\mu_{ijkl}^{0}\overline{\sigma}_{kl}+\sum_{k=1}^{K}\rho_{k}^{(A)}(\mathbf{Q}^{(k)}\overline{\sigma})_{ij}\mu_{ijkl}^{0}(\mathbf{Q}^{(k)}\overline{\sigma})_{kl}\right\} 
\end{gather*}
and the density of elastic energy for the infinite body
\begin{gather*}
e\equiv e(\sigma;\,\boldsymbol{p})=\\
=\underset{|A|\rightarrow\infty}{\lim}e_{A}(\sigma;\boldsymbol{p}^{(A)})=\left[1-\sum_{k=1}^{K}\rho_{k}^{(A)}\right]\tilde{e}(\overline{\sigma})+\sum_{k=1}^{K}\rho_{k}^{(A)}\tilde{e}(\mathbf{Q}^{(k)}\overline{\sigma})
\end{gather*}
 The elastic energy $E_{k,A}(\sigma;\,\boldsymbol{p}^{(A)})$ added
by a $k$-type crack has the form 
\begin{gather*}
E_{k,A}(\sigma;\,\boldsymbol{p}^{(A)})=a^{3}|A|\left[e_{A}(\sigma;\, p_{1}^{(A)},\ldots,p_{k}^{(A)}+\frac{1}{|A|},...,p_{K}^{(A)})-e_{A}(\sigma;\,\boldsymbol{p}^{(A)})\right]=\\
=a^{3}\frac{\partial e_{A}(\sigma;\,,\boldsymbol{p}^{(A)})}{\partial p_{k}^{(A)}}+O(|A|^{-1})
\end{gather*}
 and for the infinite lattice 
\[
E_{k}(\sigma;\,\boldsymbol{p})=a^{3}e_{k}=a^{3}\frac{\partial e}{\partial p_{k}},
\]
where $e_{i}\equiv e_{i}(\sigma;\, p_{t}(j),\, j=1,...,I)$ is the
density of elastic energy in the cell with a $i$-type crack. To express
intensity (\ref{eq:3}) explicitly we define the energy of microdefects
generating cracks in a similar form

\begin{eqnarray*}
E_{i}^{0}(\sigma;\,\boldsymbol{p})=a_{0}^{3}e_{i}
\end{eqnarray*}
 where $a_{0}$ is of the order of intermolecular distance.

The density of elastic energy can be expressed in the form

\begin{eqnarray}
e & = & \frac{1}{2}\mu_{ijmn}(\sigma;\,\boldsymbol{p})\sigma_{ij}\sigma_{mn}\equiv\frac{1}{2}(\sigma\mathbf{:}\mu\sigma),\label{eq:6}
\end{eqnarray}
 where $\mu$ is the effective elastic pliability tensor, $\mu_{ijmn}\equiv\mu_{ijmn}(\sigma;\,\boldsymbol{p})$.
As the component $s_{11}$ of the operator $\mathtt{\boldsymbol{\mathbb{\mathbb{\mathbf{S}}}}}$
has the jump when the normal stress applied on the plane of the crack
changes the sign, the dependence $\overline{\sigma}$ on $\sigma$
is piecewise linear and $\mu_{ijmn}$ are step functions of $\sigma.$ 

Similarly, 
\begin{align}
e_{k} & =\frac{1}{2}\mu_{klmn}^{k}\sigma_{kl}\sigma_{mn}\equiv\frac{1}{2}(\sigma:\mu^{k}\sigma),\label{eq:7}
\end{align}
 where tensor components $\mu_{ijmn}^{k}\equiv\mu_{ijmn}^{k}(\sigma;\,\boldsymbol{p})$
are step functions of $\sigma.$

\section{EQUATIONS OF NON-EQUILIBRIUM THERMODYNAMICS }

Consider deformation in a body under stress. A tensor $u_{ij}$ of
the total macroscopic strain consists of the tensor of reversible
(it vanishes if the stress is zero) elastic strain and the tensor
of irreversible residual strain $r_{ij}$=$u{}_{ij}-\varepsilon_{ij}$
appearing when microcracks are healing.

The total macroscopic strain was defined previously \citet{Gertzik 1998}
by 
\begin{eqnarray*}
u_{ij}=\underset{n\rightarrow\infty}{\lim}\frac{1}{|A_{n}|a^{3}}\underset{S}{\int}\,\frac{1}{2}(X_{i}n_{j}+X_{j}n_{i})ds
\end{eqnarray*}
where $A_{n}$ are cubes with centers at origin, $A_{n}\subseteq A_{n+1}$,
$\cup_{n=1}^{\infty}A_{n}=\mathbb{Z}^{\nu},$ $X_{i}$ are displacement
components, $n_{i}$ are normal components of the surface, and the
integral is taken over the surface $S$ of the cube. It has been shown
(see (\ref{eq:6})) that macroscopic elastic strain $\varepsilon_{ij}$
has the form 
\begin{gather}
\varepsilon_{ij}=\frac{\partial e}{\partial\sigma_{ij}}=\mu_{ijmn}\sigma_{mn}\label{eq:8}
\end{gather}
 if the crack configuration is fixed.

During the time interval $dt$ the density of $k$-type crack increases
by $\left[1-\sum_{l=1}^{K}p_{l}\right]\mathrm{\Lambda_{\mathit{k}}}dt$
and the strain increment $du_{ij}^{k(+)}$, in accordance with (\ref{eq:7}),
has the form 
\begin{gather}
du_{ij}^{k(+)}=\mu_{ijmn}^{k}\sigma_{mn}\left[1-\sum_{l=1}^{K}p_{l}(t)\right]\mathrm{\Lambda_{\mathit{k}}}dt\label{eq:11-1}
\end{gather}
We remind that the stress tensor $\sigma_{ij}$ can be expressed as
the sum of two other stress tensors: the mean hydrostatic stress tensor
$\bar{\sigma}\delta_{ij}$ which tends to change the volume of the
stressed body, and the deviatoric component called the stress deviator
tensor, $s_{ij}$ which tends to distort it: $\sigma_{ij}=s_{ij}+\bar{\sigma}\delta_{ij}$
where $\bar{\sigma}$ is the mean stress given by $\bar{\sigma}=\frac{\sigma_{ii}}{3}$.
After some crack has been healed, no work is done by the stress deviator
and the strain deviator does not change (the displaced sides of a
closed crack stick together and the strain caused by it becomes residual
instead of elastic strain) but every diagonal element of the strain
tensor $u_{ij}$ decreases by 1/3 of volume strain due to the crack
(as the density of the material remains constant, the volume strain
due to the open crack disappears), that is, the strain decrement $du_{ij}^{k(-)}$,
in accordance with (\ref{eq:7}) is of the form

\begin{eqnarray*}
du_{ij}^{k(-)}=\tilde{\mu}_{ijmn}^{k}\sigma_{mn}p_{k}(t)\mathrm{M_{\mathit{k}}}dt
\end{eqnarray*}
 where

\[
\tilde{\mu}_{ijmn}^{k}=\begin{cases}
\delta_{ij}\mu_{ijmn}^{k}/3,\:\mathrm{if\: crack\: is\: open,}\\
0,\:\mathrm{if\: crack\: is\: closed}.
\end{cases}
\]

The increment $du_{ij}\vert_{\sigma}$ of the strain $\varepsilon_{ij}$
during time interval $dt$ for a constant $\sigma$ is represented
by
\begin{gather*}
du_{ij}\vert_{\sigma}=\sum_{k=1}^{K}[du_{ij}^{k(+)}-du_{ij}^{k(-)}]=\\
=\sum_{k=1}^{K}\left(\mu_{ijmn}^{k}\sigma_{mn}\left[1-\sum_{l=1}^{K}p_{l}(t)\right]\mathrm{\Lambda_{\mathit{k}}}-\tilde{\mu}_{ijmn}^{k}\sigma_{mn}p_{k}(t)\mathrm{M_{\mathit{k}}}\right)dt
\end{gather*}
 The last expression and (\ref{eq:8}) yield the total strain increment
\begin{gather}
du=\mu d\sigma+\sum_{k=1}^{K}\left(\left[1-\sum_{l=1}^{K}p_{l}(t)\right]\mathrm{\Lambda_{\mathit{k}}}\mu^{k}-p_{k}(t)\mathrm{M_{\mathit{k}}}\tilde{\mu}^{k}\right)\sigma dt\label{eq:yy}
\end{gather}
 which shows that appearance and disappearance of cracks give rise
to viscous strain component. 

Using (\ref{eq:9}) and the relation

\begin{gather*}
d\mu_{ijmn}=d\left(\frac{\partial e}{\partial\sigma_{ij}\partial\sigma_{mn}}\right)=\frac{\partial^{2}}{\partial\sigma_{ij}\partial\sigma_{mn}}\sum_{k=1}^{K}\frac{\partial e}{\partial p_{k}}dp_{k}=\\
=\sum_{k=1}^{K}\frac{\partial^{2}}{\partial\sigma_{ij}\partial\sigma_{mn}}e_{k}dp_{ik}=\sum_{k=1}^{K}\mu_{ijmn}^{k}dp_{k},
\end{gather*}
 i.e.
\begin{eqnarray}
d\mu=\sum_{k=1}^{K}\mu^{k}\left(\left[1-\sum_{l=1}^{K}p_{l}(t)\right]\mathrm{\Lambda_{\mathit{k}}}-p_{k}(t)\mathrm{M_{\mathit{k}}}\right)dt\label{eq:try}
\end{eqnarray}
we have

\[
d\varepsilon=\mu d\sigma+\sum_{k=1}^{K}\mu^{k}\left(\left[1-\sum_{l=1}^{K}p_{l}(t)\right]\mathrm{\Lambda_{\mathit{k}}}-p_{k}(t)\mathrm{M_{\mathit{k}}}\right)\sigma dt
\]
and

\[
dr=\sum_{k=1}^{K}\left(\mu^{k}-\tilde{\mu}^{k}\right)p_{k}(t)\mathrm{M_{\mathit{k}}}\sigma dt.
\]

When some procedure of strain change is chosen as a prescribed external
condition, then the last formula leads to the equation for external
stress 
\begin{gather}
d\sigma=\lambda\left\{ du-\sum_{k=1}^{K}\left(\left[1-\sum_{l=1}^{K}p_{l}(t)\right]\mathrm{\Lambda_{\mathit{k}}}\mu^{k}-p_{k}(t)\mathrm{M_{\mathit{k}}}\tilde{\mu}^{k}\right)\sigma dt\right\} \label{eq:10-1}
\end{gather}
 where $\lambda$ is the inverse of $\mu$ the effective stiffness
tensor (in general, anisotropic) for a body with cracks.

The crack surface energy is represented by $\pi r^{2}\gamma,$ where
$\gamma$ is the density of surface energy. The density of internal
energy for the body (apart from an additive constant independent of
external conditions and crack densities) is the sum of the density
of elastic energy, surface crack energy, and thermal (vibrational)
energy of atoms. The first law of thermodynamics can be written in
that case, in view of (\ref{eq:6}), as 
\begin{gather*}
\frac{1}{2}d(\sigma\cdot\mu\sigma)+\gamma\frac{\pi r^{2}}{a^{3}}\sum_{k=1}^{K}\left(\left[1-\sum_{l=1}^{K}p_{l}(t)\right]\mathrm{\Lambda_{\mathit{k}}}-p_{k}(t)\mathrm{M_{\mathit{k}}}\right)dt+\rho C_{p}dT=\\
=(\sigma\cdot du)+dQ
\end{gather*}
 where $\rho$ is the density of the material, $C_{p}$ is specific
heat at constant pressure, and $dQ$ is the specific heat energy increment
from the outside. Using (\ref{eq:yy},\ref{eq:9},\ref{eq:try}) and
the equality $\mu_{ijmn}=\mu_{mnij}$ we get the expression for the
temperature increment: 
\begin{gather}
\rho C_{p}dT=dQ+\sum_{k=1}^{K}\left[\frac{1}{2}(\sigma\cdot\mu^{k}\sigma)-\gamma\frac{\pi r_{i}^{2}}{a^{3}}\right]\left[1-\sum_{l=1}^{K}p_{l}(t)\right]\mathrm{\Lambda_{\mathit{k}}\mathit{dt}}+\nonumber \\
+\sum_{k=1}^{K}\biggr[\frac{1}{2}(\sigma\cdot\mu^{k}\sigma)+\gamma\frac{\pi r_{i}^{2}}{a^{3}}-(\sigma\cdot\tilde{\mu}^{k}\sigma)\biggr]p_{k}(t)\mathrm{M_{\mathit{k}}\mathit{dt}}\label{eq:8-1-1}
\end{gather}
Equations (\ref{eq:55}), (\ref{eq:9}), (\ref{eq:10-1}) and (\ref{eq:8-1-1})
with (\ref{eq:2}), (\ref{eq:3}) constitute a closed system describing
the non-equilibrium thermodynamics of a deformation process in a solid
body (as usual, superior dots denote differentiation with respect
to time):

\begin{equation}
\frac{dp_{k}(t)}{dt}=\left[1-\sum_{i=1}^{K}p_{l}(t)\right]\mathrm{\Lambda_{\mathit{k}}}-p_{k}(t)\mathrm{M_{\mathit{k}}},\: k=1,...,K\label{eq:12}
\end{equation}

\begin{equation}
\dot{\sigma}=\lambda\left\{ \dot{u}-\sum_{k=1}^{K}\left(\left[1-\sum_{l=1}^{K}p_{l}(t)\right]\mathrm{\Lambda_{\mathit{k}}}\mu^{k}-p_{k}(t)\mathrm{M_{\mathit{k}}}\tilde{\mu}^{k}\right)\sigma\right\} \label{eq:13}
\end{equation}
\begin{gather}
\rho C_{p}\dot{T}=\left[1-\sum_{l=1}^{K}p_{l}(t)\right]\mathrm{\Lambda_{\mathit{k}}}\sum_{k=1}^{K}\left[\frac{1}{2}(\sigma\cdot\mu^{k}\sigma)-\gamma\frac{\pi r_{i}^{2}}{a^{3}}\right]+\nonumber \\
+\sum_{k=1}^{K}\biggr[\frac{1}{2}(\sigma\cdot\mu^{k}\sigma)+\gamma\frac{\pi r_{i}^{2}}{a^{3}}-(\sigma\cdot\tilde{\mu}^{k}\sigma)\biggr]p_{k}(t)\mathrm{M_{\mathit{k}}+G}\label{eq:dgf}
\end{gather}
 where $G$ is the heat influx rate per unit volume. It is necessary
to notice that any nonzero stress does a work that dissipates irreversibly
and the system is not equilibrium even in the stationary case.

\section{POWER OF ACOUSTIC EMISSION}

According to (\ref{eq:7}), (\ref{eq:11-1}) the external forces $\sigma$
make at occurrence of cracks the differential work

\begin{gather*}
\sum_{k=1}^{K}\sigma_{ij}du_{ij}^{k(+)}=\left[1-\sum_{k=1}^{K}p_{k}(t)\right]\sum_{k=1}^{K}\sigma_{ij}\mu_{ijmn}^{k}\sigma_{mn}\mathrm{\Lambda_{\mathit{k}}}dt=\\
=\left[1-\sum_{k=1}^{K}p_{k}(t)\right]\sum_{k=1}^{K}(\sigma:\mu^{k}\sigma)\mathrm{\Lambda_{\mathit{k}}}dt.
\end{gather*}
and the increment of density of elastic energy due to new microcracks
is equal to $\frac{1}{2}\left[1-\sum_{k=1}^{K}p_{k}(t)\right]\sum_{k=1}^{K}(\sigma:\mu^{k}\sigma)\mathrm{\Lambda_{\mathit{k}}}dt$.
Thus only half of the work goes elastic energy and the second half
dissipates through acoustic emission and subsequently becomes heat.
This is caused by spasmodic changes in the effective elastic moduli
of the cells where cracks appear. Therefore the power of acoustic
emission per unit volume $w$ is

\section*{{\normalsize{
\begin{gather}
w=\frac{1}{2}\left[1-\sum_{k=1}^{K}p_{k}(t)\right]\sum_{k=1}^{K}(\sigma:\mu^{k}\sigma)\mathrm{\Lambda_{\mathit{k}}.}\label{eq:er}
\end{gather}
}}}

\section{DYNAMICS OF A CONTINUUM }

To derive dynamic equations for continuum from equations for a material
point we introduce velocities $\boldsymbol{v}\equiv v_{k}(\boldsymbol{x},t)$,
$k=1$, $2$, $3$ in three-dimensional space, $\boldsymbol{x}=(x_{1},x_{2},x_{3})$.
Euler vector coordinates $\boldsymbol{x}\equiv\boldsymbol{x}(\boldsymbol{X},t)$
are functions of Lagrange (material) coordinates $\boldsymbol{X}$
of a continuum point in some initial configuration (see, e.g., \citet{Day})
and the velocity of the point is $\boldsymbol{v}\equiv\boldsymbol{v}(\boldsymbol{X},t)=\frac{\partial\boldsymbol{x}}{\partial t}$
. Let $f$ be a local parameter of the continuum. This can be considered
to be both $f(\boldsymbol{X},t)$ and $f(\boldsymbol{\boldsymbol{x}},t)$.

The full derivative of $f$ with respect to time $\dot{f}\equiv\partial f(\boldsymbol{X},t)/\partial t$
is equal to

\[
\dot{f}\equiv\partial f(\boldsymbol{x}(\boldsymbol{X},t),t)/\partial t+\sum_{i=1}^{3}\frac{\partial f(\boldsymbol{x},t)}{\partial x{}_{i}}\frac{\partial x_{i}}{\partial t}=\frac{\partial f(\boldsymbol{x},t)}{\partial t}+(\boldsymbol{v}\cdot\nabla f(\boldsymbol{x},t))
\]
for any differentiable function $f$. 

To calculate the strain rate we use the small strain approximation
(the finite strain theory is too bulky to develop it here). In this
approximation the Lagrange and Euler coordinates are identical and
the small displacement vector $U_{i}(x,t)$,$\, i=1,2,3,$ is defined
at each point. Symmetrized velocity gradient $\nabla^{(s)}\boldsymbol{v}=\frac{1}{2}(\frac{\partial v_{k}}{\partial x_{l}}+\frac{\partial v_{l}}{\partial x_{k}})$
is by definition $u_{lk}=\frac{1}{2}\left(\frac{\partial U{}_{l}}{\partial x_{k}}+\frac{\partial\boldsymbol{U}_{k}}{\partial x_{l}}\right)$
equal to the strain rate tensor $\nabla^{(s)}\boldsymbol{v}=\frac{\partial}{\partial t}u$.

The law of conservation of momenum (Newton's second law) takes the
form 
\begin{eqnarray*}
\rho\frac{\partial\boldsymbol{v}}{\partial t}=-\rho(\boldsymbol{v}\cdot\nabla\boldsymbol{v}){\rm +div}\;\sigma+\rho F
\end{eqnarray*}
 where $\rho$ is the density and $F$ is the force per unit mass.
Using (\ref{eq:12}) we get a kinetic equation for crack densities
\[
\frac{\partial p_{k}}{\partial t}=-(\boldsymbol{v}\cdot\nabla p_{k})+\left[1-\sum_{i=1}^{K}p_{l}(t)\right]\mathrm{\Lambda_{\mathit{k}}}-p_{k}\mathrm{M_{\mathit{k}}},\: k=1,...,K
\]
 Equations (\ref{eq:13}) lead to equations of state for the continuum:
\begin{eqnarray*}
\frac{\partial\sigma}{\partial t} & = & -(\boldsymbol{v}\cdot\nabla\sigma)+\lambda\biggr\{\nabla^{(s)}\boldsymbol{v}-\sum_{k=1}^{K}\left(\left[1-\sum_{l=1}^{K}p_{l}\right]\mathrm{\Lambda_{\mathit{k}}}\mu^{k}-p_{k}M_{k}\tilde{\mu}^{k}\right)\biggr\}
\end{eqnarray*}
Lastly, we get the equation of heat conduction with heat sources using
the law of conservation of energy in the form (\ref{eq:dgf}) and
the Fourier law $q=\kappa\nabla T$ where $q$ is the heat flux vector,
$G={\rm div}\; q,\;\kappa$ is thermal conductivity:
\begin{gather*}
\frac{\partial T}{\partial t}=-(\boldsymbol{v}\cdot\nabla T)+(\rho C_{p})^{-1}\left\{ \left[1-\sum_{l=1}^{K}p_{l}(t)\right]\mathrm{\Lambda_{\mathit{k}}}\sum_{k=1}^{K}\left[\frac{1}{2}(\sigma\cdot\mu^{k}\sigma)-\gamma\frac{\pi r_{i}^{2}}{a^{3}}\right]+\right.\\
+\left.\sum_{k=1}^{K}\biggr[\frac{1}{2}(\sigma\cdot\mu^{k}\sigma)+\gamma\frac{\pi r_{i}^{2}}{a^{3}}-(\sigma\cdot\tilde{\mu}^{k}\sigma)\biggr]p_{k}(t)\mathrm{M_{\mathit{k}}}\right\} +\chi\Delta T,
\end{gather*}
 $\chi=\kappa(\rho C_{p})^{-1}$ is thermal diffusivity.

\section{PHASE TRANSITION}

Let $\sigma_{12}=\sigma_{21}\geq0$ be external shear stress applied
to the solid, $\sigma=\sigma_{12}/\sqrt{\mu_{s}},$ $\mu_{s}$ is
a the shear modulus, $U(\sigma;\,\boldsymbol{p}\mathrm{(}t\mathrm{)})=U$,
$U$ is a constant, and allowed cracks have their normal parallel
to axes $x_{1}$ and $x_{2}$. Then (\ref{eq:2}) - (\ref{eq:4})
give 

\begin{gather*}
\overline{\sigma}=\dfrac{\sigma}{1-\theta p},\\
\Lambda_{\mathit{\mathrm{1}}}=\Lambda_{2}=\Lambda=c_{1}\exp\left\{ \beta\left[a_{0}\left(\dfrac{\sigma}{1-\theta p}\right)^{2}-H\right]\right\} ,\\
\mathrm{M{}_{\mathit{\mathrm{1}}}=M_{2}=M=c_{0}\exp\left\{ -\beta U,\right\} }
\end{gather*}
and the two types of cracks are indistinguishable in (\ref{eq:12}).
We introduce their total density $p$, and get from (\ref{eq:12})
the equation for stationary case
\begin{gather*}
c_{1}\left(1-p\right)\;\exp\left\{ \beta\left[a_{0}\left(\dfrac{\sigma}{1-\theta p}\right)^{2}-H\right]\right\} -c_{0}p\;\exp\left\{ -\beta U\right\} =0.
\end{gather*}
From this we get the dependence of $\sigma$ on $\mathit{\lyxmathsym{\textcyr{\char240}}}$
and $\beta$

\begin{equation}
\sigma=\sqrt{\frac{1}{a_{0}\beta}(1-\theta p)^{2}\left[\ln\frac{p}{1-p}-\ln\frac{c_{1}}{c_{0}}+\beta(H-U)\right]}\label{eq:66}
\end{equation}
 and $\sigma\geq0$ when 
\[
\frac{c_{1}\exp\left\{ -\beta H\right\} }{c_{1}\exp\left\{ -\beta H\right\} +c_{0}\exp\left\{ -\beta U\right\} }\leq p<1.
\]
As $\sigma$ changes from 0 to $\infty,$ the phase transition takes
place if $\min\frac{d\sigma}{dp}\leq0$. At $\min\frac{d\sigma}{dp}<0$
the dependence of $p$ on $\mathit{\sigma}$ has an $S$-shaped form
(see Fig.1) similar to the van der Waals curve.

\medskip{}

\begin{center}
\begin{figure}[H]
\begin{centering}
\includegraphics[scale=0.35]{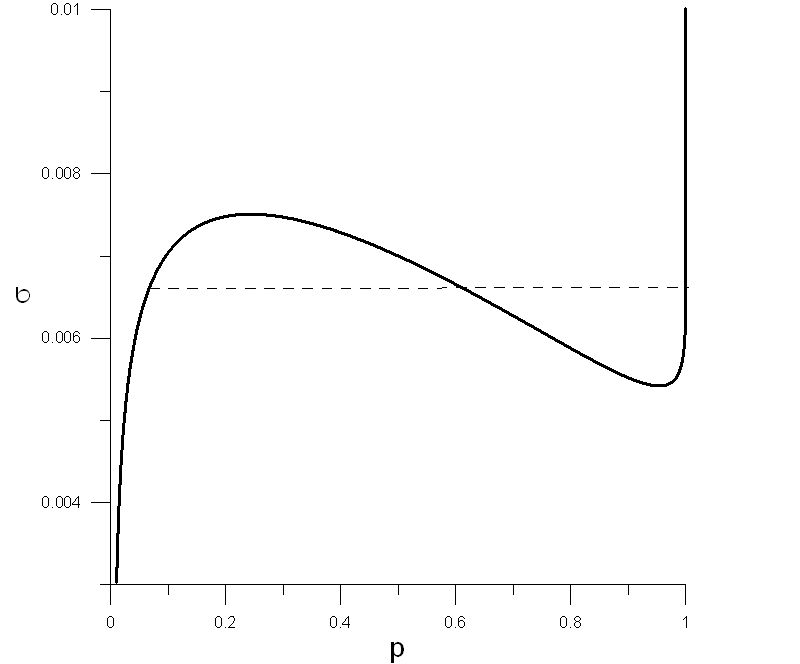}
\par\end{centering}

\centering{}\textbf{\large{Fig. 1}}
\end{figure}

\par\end{center}

\bigskip{}
Following the technique of pioneer work of M.Kac for the Curie-Weiss
model \citet{Kac} we construct from (A.3), (A.4) (see Appendix)the
asymptotic distribution $u_{N}(x)$ of $x=\frac{n}{N}$ at large $N$:

\begin{gather}
\mathrm{\mathit{u}_{\mathit{N}}(}x\mathrm{)\equiv}\mathrm{Pr\{}\mathit{n=xN}\mathrm{\}}=\sum_{A:|A|=xN}P\mathrm{_{\mathit{N}}(}A\mathrm{)\sim}\nonumber \\
\sim\frac{N\mathrm{!}}{\mathrm{(\mathrm{\mathit{N}}}-xN\mathrm{)!}(xN)\mathrm{!}}\frac{1}{Z_{N}\sqrt{\Lambda(y)\mathrm{M}(y)}}\exp\left\{ N\int_{0}^{x}\ln f(z)dz\right\} \sim\nonumber \\
\mathrm{\mathit{\sim}\frac{1}{\mathit{\mathit{Z}}_{\mathit{N}}}\sqrt{\frac{1}{2\pi\mathit{N}(1-\mathit{x})\mathit{x\mathrm{\mathrm{\Lambda}(}y\mathrm{)}\mathrm{M}\mathrm{(}y\mathrm{)}}}}\exp\left\{ \mathit{NF\mathrm{(}x\mathrm{)}}\right\} ,}\label{eq:rrq}
\end{gather}

\[
F(x)=x\ln\frac{c_{1}}{c_{0}}+x\beta\left(U-H+a_{0}\dfrac{\sigma^{2}}{1-\theta x}\right)-x\ln x-(1-x)\ln(1-x),
\]

\begin{gather*}
Z_{N}\sim\int_{0}^{1}\sqrt{\frac{1}{2\pi N(1-x)x\Lambda(y)\mathrm{M}(y)}}\exp\left\{ NF(\boldsymbol{x})\right\} dx.
\end{gather*}

As $N\rightarrow\infty$ the probability measure tends to a measure
concentrated at a points of a strict maximum of $F(x)$ and a first-order
phase transition takes place if $F(x)$ has two identical maxima.
The condition of extremum of $F(x)$ 

\[
\frac{dF}{dx}=\ln\frac{c_{1}}{c_{0}}+\beta\left(U-H\right)+a_{0}\beta\frac{\sigma^{2}}{(1-\theta x)^{2}}-\ln\frac{x}{1-x}=0
\]
is equivalent to (\ref{eq:66}). The requirement of equality of two
maxima $F(x_{1})=F(x_{2})$ of function $F(x)$ yields an analog of
Maxwell equal area rule 
\[
\int_{x_{1}}^{x_{2}}\left[\ln\frac{c_{1}}{c_{0}}+\beta_{c}\left(U-H\right)+a_{0}\beta\frac{\sigma_{c}^{2}}{(1-\theta z)^{2}}-\ln\frac{z}{1-z}\right]dz=0.
\]
.

A second-order phase transition takes place at critical point $(p_{c},\beta_{c},\sigma_{c})$
where $\frac{dF}{dx}$$\mid_{p_{c}}=\frac{d^{2}F}{dx^{2}}$$\mid_{p_{c}}=\frac{d^{3}F}{dx^{3}}\mid_{p_{c}}=0.$
From

\[
\frac{dF}{dx}\mid_{p_{c}}=\ln\frac{c_{1}}{c_{0}}+\beta_{c}\left(U-H\right)+a_{0}\beta\frac{\sigma_{c}^{2}}{(1-\theta p_{c})^{2}}-\ln\frac{p_{c}}{1-p_{c}}=0,
\]

\[
\frac{d^{2}F}{dx^{2}}\mid_{p_{c}}=2a_{0}\beta_{c}\frac{\theta\sigma_{c}^{2}}{(1-\theta p_{c})^{3}}-\frac{1}{p_{c}\left(1-p_{c}\right)}=0,
\]

\[
\frac{d^{3}F}{dx^{3}}\mid_{p_{c}}=6a_{0}\beta_{c}\frac{\theta^{2}\sigma_{c}^{2}}{(1-\theta p_{c})^{4}}+\frac{1-2p_{c}}{p_{c}^{2}\left(1-p_{c}\right)^{2}}=\frac{1-2(1-\theta)p_{c}-\theta p_{c}^{2}}{(1-\theta p_{c})p_{c}^{2}\left(1-p_{c}\right)^{2}}=0,
\]
we find coordinates of the critical point at $H\neq W$

\[
p_{c}=\frac{\sqrt{1-\theta+\theta^{2}}-1+\theta}{\theta},
\]

\begin{gather*}
\beta_{c}=\frac{1}{H-U}\left[\ln\frac{c_{1}}{c_{0}}+\frac{\left(2-\theta-\sqrt{1-\theta+\theta^{2}}\right)\theta}{(2-\theta)\sqrt{1-\theta+\theta^{2}}-2(1-\theta)-\theta^{2}}-\right.\\
\left.-\ln\frac{\sqrt{1-\theta+\theta^{2}}-1+\theta}{1-\sqrt{1-\theta+\theta^{2}}}\right],
\end{gather*}

\[
\sigma_{c}=\sqrt{\frac{1}{a_{0}\beta_{c}}\frac{(1-\theta+\theta^{2})\left(2-\theta-\sqrt{1-\theta+\theta^{2}}\right)\theta}{(2-\theta)\sqrt{1-\theta+\theta^{2}}-2(1-\theta)-\theta^{2}}}.
\]
We derive the same results from (\ref{eq:66}) with $(p,\beta,\sigma)=(p_{c},\beta_{c},\sigma_{c})$
and $\frac{d\sigma}{dx}$$\mid_{p_{c}}=\frac{d^{2}\sigma}{dx^{2}}$$\mid_{p_{c}}=0.$

In the case $H=U$ we have

\[
p_{c}=\frac{\sqrt{1-\theta+\theta^{2}}-1+\theta}{\theta},
\]

\[
\beta_{c}\sigma_{c}^{2}=\frac{(1-\theta p_{c})^{2}}{a_{0}}\left(\ln\frac{p_{c}}{1-p_{c}}-\ln\frac{c_{1}}{c_{0}}\right).
\]
In this case there is the line of critical points $\beta_{c}\sigma_{c}^{2}=const$
on the plane $(\beta,\sigma)$ instead of a single critical point
in the ordinary case.

\section{DISTRIBUTION OF LOGARITHMIC POWER OF ACOUSTIC EMISSION AT THE CRITICAL
POINT FOR A LARGE SYSTEM}

In the conditions of the previous section the power of acoustic emission
$W$ for a solid with $N$ cells has the form (\ref{eq:er})

\[
W=\frac{N}{2}(1-p)\frac{\theta\sigma^{2}}{(1-\theta p)^{2}}c_{1}\exp\left\{ \beta\left[a_{0}\left(\dfrac{\sigma}{1-\theta p}\right)^{2}-H\right]\right\} .
\]
 To find the distribution $\varphi_{N}(y)$ of $y=\ln W$ at the critical
point we notice that the distribution $\mathit{u_{N}}(x\mathrm{)}$
of crack density $x=\frac{n}{N}$ at large $N$, according to (\ref{eq:rrq}),
looks at the critical point like

\begin{gather*}
\mathit{u_{N}}(\mathit{x}\mathrm{)}\sim\frac{1}{\mathit{\mathit{Z}}_{\mathit{N}}}\sqrt{\frac{1}{2\pi\mathit{N}(1-\mathit{x})\mathit{x\mathrm{\Lambda(\mathit{y})\mathrm{M}(\mathit{y})}}}}\exp\left\{ \mathit{NF\mathrm{(}x\mathrm{)}}\right\} \sim\\
\left.\sim\mathit{u_{N}}(\mathit{p_{c}}\mathrm{)}\exp\left\{ \mathit{\frac{N}{\left(2n\right)\mathrm{!}}\frac{d^{2\mathit{\mathit{n}}}F_{c}}{dx^{\mathit{\mathrm{2}\mathit{n}}}}\mid_{p_{c}}\left(x-p_{c}\right)^{\mathrm{2}\mathit{n}}}-\mathit{\frac{\mathrm{1}}{\mathrm{2}}\frac{d}{dx}}\ln\left[\Lambda(x)\mathrm{M}(x)(1-\mathit{x})\mathit{x}\right]\mid_{p_{c}}\left(\mathit{x-p_{c}}\right)\right\} \right\} ,
\end{gather*}
 at the critical point where $n$ is the smallest positive number
at which the derivative $\frac{d^{2\mathit{\mathit{n}}}F_{c}}{dx^{\mathit{\mathrm{2}\mathit{n}}}}$
is distinct from 0 ($\frac{d^{2\mathit{\mathit{n-\mathrm{1}}}}F_{c}}{dx^{\mathit{\mathrm{2}\mathit{n}}-1}}\mid_{p_{c}}=0$
because $\mathit{u_{N}}(\mathit{x}\mathrm{)}$ has a maximum at the
point $p_{c})$. As

\[
\frac{d^{4}F_{c}}{dx^{4}}\mid_{p_{c}}=\frac{-2(1-\theta+\theta p_{c})}{(1-\theta p_{c})p_{c}^{2}\left(1-p_{c}\right)^{2}}\leq0,
\]
 we have

\begin{gather*}
\\
\mathrm{\mathit{u_{N}}(\mathit{x}\mathrm{)}\sim\frac{1}{\mathit{\mathit{Z}}_{\mathit{N}}}\,\exp\left\{ \mathit{NF_{c}\mathrm{(}p_{c}\mathrm{)}}\right\} \exp\left\{ \mathit{\frac{N}{\mathrm{4!}}\frac{d^{\mathrm{4}}F_{c}}{dx^{\mathit{\mathrm{4}}}}\mid_{p_{c}}\left(x-p_{c}\right)^{\mathrm{4}}-\mathit{\frac{\mathrm{1}}{\mathrm{2}}\frac{d}{dx}}\left[\mathrm{\ln\Lambda(\mathit{x})\mathrm{M}(\mathit{x})(1-\mathit{x})\mathit{x}}\right]\mid_{p_{c}}\left(\mathit{x-p_{c}}\right)}\right\} }.
\end{gather*}
 The distribution $\varphi_{N}(y)$ is expressed by $\mathit{u_{N}}(\mathit{x}\mathrm{)}$
as 
\[
\varphi_{N}(y)=\sum_{x:\ln W(x)=y}u(x(y))\left(\frac{dy}{dx}\right)^{-1}=\sum_{x:W(x)=y}u(x(y))W\left(\frac{dW}{dx}\right)^{-1},
\]
 and

\[
\frac{dW}{dx}=W\left[\frac{2\theta}{(1-\theta x)}-\frac{1}{1-x}+\frac{2\alpha\beta\sigma^{2}\theta}{(1-\theta x)^{3}}\right]\equiv Wg(x).
\]
 Therefore $\varphi_{N}(y)=\sum_{x:W(x)=y}\frac{u(x(y))}{g(x)}$ and
at the critical point

\begin{gather*}
\varphi_{N}(y)=\sum_{x:W(x)=y}\frac{u(x(y))}{g(x)}\sim\left\{ C\exp\frac{N}{\mathrm{4!}}\frac{d^{4}F_{c}}{dx^{4}}\mid_{p_{c}}\frac{\mathrm{1}}{g\mathrm{(}p_{c}\mathrm{)^{4}}}\left(y-\ln W_{c}\right)^{4}-\right.\\
\left.-\mathit{\frac{\mathrm{1}}{\mathrm{2}}\frac{d}{dx}}\mathrm{\ln\left[g(x)\Lambda(\mathit{x})\mathrm{M}(\mathit{x})(1-\mathit{x})\mathit{x}\right]}\mid_{p_{c}}\frac{\mathrm{1}}{g\mathrm{(}p_{c}\mathrm{)}}\left(y-\ln W_{c}\right)\right\} ,
\end{gather*}

As the energy $\Delta E$ emitted during a small time $\Delta t$
is $W\Delta t$, the logarithm of its distribution has the same form
as $\ln\varphi_{N}$ 

\begin{equation}
\ln\varphi_{N}=a-b\ln W-c\left(\ln W-\ln W_{c}\right)^{4}.\label{eq:rr}
\end{equation}
This expression is surprisingly similar to the empirical distribution
of earthquakes energy. Theoretically it satisfies to the Gutenberg-Richter
law \citet{Gutenberg} $\ln\frac{N_{E}}{N_{total}}=a-b\ln E$, $N_{E}$
is the number of earthquakes with energy $E$, $N_{total}$ is total
number of earthquakes, but empirical curves deviate from linear law
downwards at both edges just as in (\ref{eq:rr}). Fig. 2 shows example
of a plot of $\varphi_{N}$ , and on Fig. 3 we see the top part of
the plot 0f $\ln\varphi_{N}$ and corresponding linear dependence.\medskip{}

\begin{center}
\includegraphics[scale=0.35]{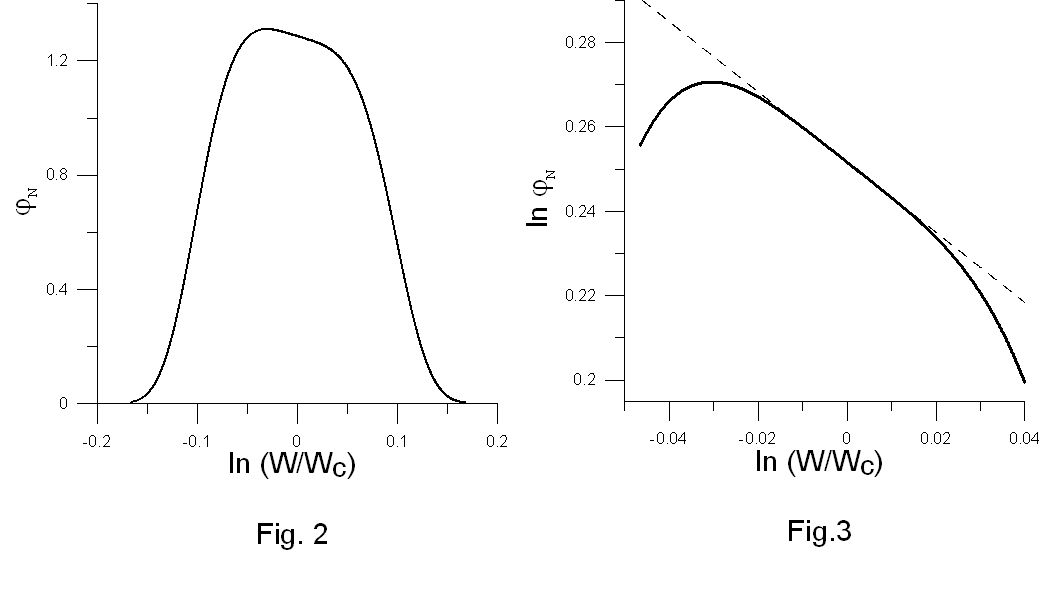}
\par\end{center}

If we assume that the area of the stress reduction is not less then
the sphere inscribed in a cubic cell, than $\frac{\pi}{6}\leq\theta<1.$
and $0.271\leq b<0.5$. The Gutenberg\textendash{}Richter law expresses
the relationship between magnitude $M$ and the number of earthquakes
with this $M.$ The magnitude is expressed through energy in joules
(\citet{Kasahara}) $M=\frac{2}{3}(\lg E-11.8)$. The relationship
between $\frac{2}{3}\lg W$ and $\lg\varphi_{N}$ has the form $\lg\varphi_{N}=a^{\prime}-b^{\prime}\lg W-c^{\prime}\left(\lg W-\lg W_{c}\right)^{4}$
with $0.407\leq b^{\prime}<0.75$. This range of $b^{\prime}$ is
comparable with the range $0.5\leq b^{\prime}<1.5$ observed for the
distribution of earthquake magnitude.

\paragraph*{\medskip{}
}

\section{ZHURKOV'S CURVES}

In \citet{Gertzik } it was demonstrated that the numerical solutions
of equations (\ref{eq:12})-(\ref{eq:dgf}) reproduce a number of
physical properties of solids observed in experiments. These properties
include creep, the existence of lower and upper yield points, strain
hardening, dilatancy, the growth of elastic anisotropy and drop in
the ratio of compressional to shear velocities under loading.

The experimental data underlying Zhurkov's formula (\ref{eq:0-1})
are also reproduced in numerical simulations.

The results of experiments for a tensile load $\sigma$are shown in
Fig.4 and are taken from \citet{Regel}:

\begin{figure}[H]
\includegraphics[scale=0.85]{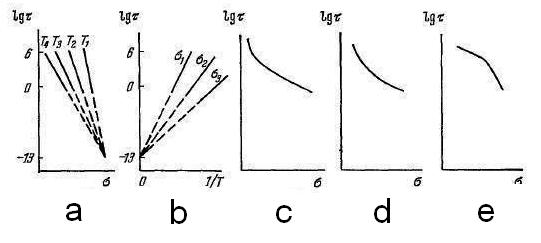}

\begin{centering}
Fig.4 
\par\end{centering}

(a) $\ln\tau$ as a function of $\sigma$ at constant $T$ ,

(b) $\ln\tau$ as a function of $T^{-1}$ at constant $\sigma$,

(c)-(e) deviations of $\ln\tau$ as a function of $\sigma$ from linearity. 
\end{figure}

\begin{flushleft}
The deviations from linearity are explained by some ``complicating
factors''. 
\par\end{flushleft}

In the conditions of two previous sections we get from (\ref{eq:12})
the equation

\[
\frac{dp}{dt}=\left(1-p\right)c_{1}\exp\left\{ \beta\left[a_{0}\left(\dfrac{\sigma}{1-\theta p}\right)^{2}-H\right]\right\} -pc_{0}\exp\left\{ -\beta U\right\} .
\]
 We assume that the lifetime $\tau$ of a specimen under tensile load
$\sigma$ at temperature $T$ is the time for which the density of
microcracks $p$ changes from 0 to the critical value $p_{0}$ (e.g.
$p_{0}=0.2$).

For different sets of constants there is a sufficiently large domain
of $\sigma$ and $T$ where $\ln\tau$ as a function of $\sigma$
and $T$ plot as straight lines similar to experimental data as shown
in Fig.5. This fact allows us to present this numeral result in the
Zhurkov form (\ref{eq:0-1}).

\begin{figure}[H]
\begin{centering}
\includegraphics[scale=0.4]{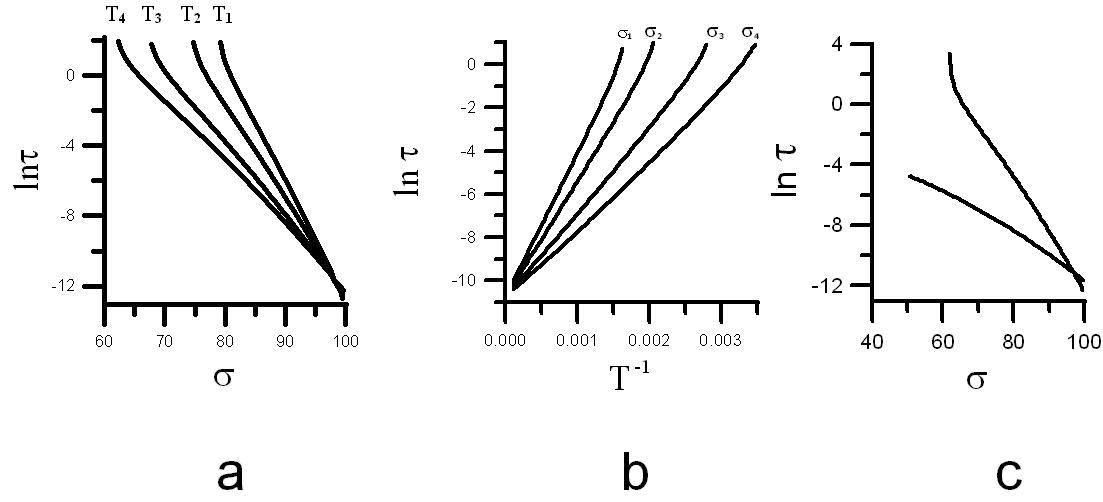} 
\par\end{centering}

\begin{centering}
Fig.5 
\par\end{centering}

(a) $\ln\tau$ as a function of $\sigma$ at constant $T$ ,

(b) $\ln\tau$ as a function of $T^{-1}$ at constant $\sigma$,

(c)-(e) deviations of $\ln\tau$ as a function of $\sigma$ from linearity.
\end{figure}

\begin{flushleft}
As the linearity and the deviations from linearity are present in
the solutions of the equation, the need to introduce ``complicating
factors'' vanishes. 
\par\end{flushleft}

\section*{Appendix}

\section*{ASYMPTOTIC STATIONARY STATE OF THE SIMPLEST MEAN FIELD MARKOV PROCESS}
\begin{defn*}
Let $\xi_{N}(t)\equiv\{\xi_{N}(x,t),$ $x\in\Omega_{N}\},\,\xi_{N}(x,t)=0,1,\,|\Omega_{N}|=N,\, t\text{\ensuremath{\ge}}0,$
($|A|$ denotes a number of elements in $A$) be $N$-component continuous-time
Markov process with state-space $\left\{ 0,1\right\} ^{\Omega_{N}}$.
Denote $n(t)=n(\xi_{N}(t))=|\{x\in\Omega_{N}:\,\xi_{N}(x,t)=1\}|$
and let $\Lambda(y),\mathrm{\, M(\mathit{y})}$ be positive continuous
functions on $[0,1]$. We assume that only one point spin-flip translations
$\left\{ 0\rightarrow1,1\rightarrow0\right\} $ are possible, with
the rates given by conditional probabilities
\[
\begin{cases}
\Pr\{\xi_{N}(x,t+h)=1|\xi_{N}(x,t)=0,n(t)\}=\Lambda\left(\frac{n(t)}{N}\right)h+o(h),\\
\Pr\{\xi_{N}(x,t+h)=0|\xi_{N}(x,t)=1,n(t)\}=\mathrm{M}\left(\frac{n(t)}{N}\right)h+o(h)\;\;\;\;(A.1)
\end{cases}
\]
This process will be named ``\emph{mean field}'' Markov process.
The random process $n(t)$ is also Markov process. For it 

\begin{gather*}
\begin{cases}
\Pr\{n(t+h)=k+1|n(t)=k\}=(N-k)\Lambda\left(\frac{k}{N}\right)h+o(h),\\
\Pr\{n(t+h)=k-1|n(t)=k\}=k\mathrm{M}\left(\frac{k}{N}\right)h+o(h),\\
\Pr\{n(t+h)=k|n(t)=k\}=1-(N-k)\Lambda\left(\frac{k}{N}\right)h-k\mathrm{M}\left(\frac{k}{N}\right)h+o(h)
\end{cases}
\end{gather*}

So for $m(t)=\frac{n(t)}{N}$ we have 
\[
\frac{d<m(t)>}{dt}=<(1-m(t))\Lambda(m(t))-m(t)\mathrm{M}(m(t))>,
\]
where <...> is the mathematical expectation. It is known (\citet{Ethier},
{[}\citet{Malyshev}) that if $N\rightarrow\infty$ and $m(0)$ tends
to a nonrandom limit $p(0)$, then $m(\text{t})$ converges in probability
to $p(\text{t})$ and $p(\text{t})$ is described by equation 
\[
\frac{dp(t)}{dt}=\left[1-p(t)\right]\mathrm{\Lambda}\left(p\mathrm{(}t\mathrm{)}\right)-p(t)\mathrm{M}\left(p(t\mathrm{)}\right)
\]
 It follows from this that if random variables $\xi_{N}(x,0)$ are
independent and equally distributed in an initial time moment, i.e.
they have the Bernoulli distribution with the parameter $p(0)$, than
in the limit $N\rightarrow\infty$ the joint distribution of values
$\xi_{N}(x,t)$ for any $t$ and any finite set of $x$ converges
to the Bernoulli distribution with the parameter $p(t)$. It is based
on the exchangeability of these random variables.
\end{defn*}
Below we'll prove that for the case of the absence of phase transitions
the components $\xi_{N}(x)$ in the stationary state are asymptotically
independent in the thermodynamic limit $N\rightarrow\infty$.

Let $A\subseteq\Omega_{N},$ and 

\begin{gather*}
P_{N}(A,t)=\Pr\{\xi_{N}(x,t)=1,\, x\in A,\,\xi_{N}(x,t)=0,\, x\in\Omega_{N}\setminus A\}.
\end{gather*}
It follows from (A.1) that

\begin{gather*}
\frac{dP_{N}(A,t)}{dt}=\sum_{x\in A}P_{N}(\mathfrak{\mathit{\mathrm{\mathit{A}}}\mathit{\setminus\left\{ x\right\} }},t)\Lambda\left(\frac{|A|-1}{N}\right)-|A|P_{N}(A,t))\mathrm{M}\left(\frac{|A|}{N}\right)+\\
+\sum_{x\in\Omega_{N}\setminus A}P_{N}(\mathfrak{\mathfrak{\mathit{\mathrm{\mathit{A}}}\mathit{\cup\left\{ x\right\} }}},t)\mathrm{M}\left(\frac{|A|+1}{N}\right)-(N-|A|)P_{N}(A,t))\Lambda\left(\frac{|A|}{N}\right).\;\;\;\;(A.2)
\end{gather*}

\smallskip{}

\begin{thm*}
If $\Lambda(y),\mathrm{\, M(y)}$ are strictly positive and differentiable,
$f\left(y\right)=\frac{\Lambda(y)}{\mathrm{M}(y)},$ function
\[
F(y)=\int_{0}^{y}\ln f(z)dz-y\ln y-(1-y)\ln(1-y)
\]
as a single maximum in $p$, $p\in(0,1)$, \textup{(}the absence of
phase transitions\textup{)} and $F^{\prime\prime}(p)$ is negative,
then in the stationary state $\xi(x)$ are independent in the limit
$N\rightarrow\infty$ and $p=\underset{N\rightarrow\infty}{\lim}\mathrm{E_{\mathit{N}}}\left\{ \frac{\mathit{n}(\xi_{N})}{N}\right\} $
is the solution of the equation 
\end{thm*}
\[
p=\frac{f(p)}{1+f(p)}.
\]
\medskip{}

\begin{proof}
Let's set 
\[
f_{N}\left(y\right)=\frac{\Lambda\left(y\right)}{\mathrm{M\left(\mathit{y+\frac{1}{N}}\right)}}.
\]
By means of direct substitution in (A.2) it is easy to check that
the solution of the system 
\begin{gather*}
\frac{dP_{N}(A,t)}{dt}=0
\end{gather*}
has a form
\end{proof}
\[
P_{N}(A)=\frac{1}{Z_{N}}\exp\left\{ \sum_{m=0}^{_{^{|A|-1}}}\ln f_{N}\left(\frac{m}{N}\right)\right\} ,
\]

\begin{gather*}
Z_{N}=\sum_{A\subseteq\Omega_{N}}\,\exp\left\{ \sum_{m=0}^{_{^{|A|-1}}}\ln f_{N}\left(\frac{m}{N}\right)\right\} =\\
=\sum_{|A|=0}^{N}\,\frac{N!}{(N-|A|)!|A|!}\exp\left\{ \sum_{m=0}^{_{^{|A|-1}}}\ln f_{N}\left(\frac{m}{N}\right)\right\} .
\end{gather*}
We have for any differentiable function $g$ 
\begin{gather*}
\sum_{m=0}^{M-1}g(\frac{m}{N})=N\sum_{m=0}^{M}g(\frac{m}{N})\frac{1}{N}-g(\frac{M}{N})=\frac{g(0)-g(\frac{M}{N})}{2}+\\
+N\sum_{m=1}^{M}\frac{g(\frac{m-1}{N})+g(\frac{m}{N})}{2}\frac{1}{N}=\\
=\frac{g(0)-g(\frac{M}{N})}{2}+N\int_{0}^{\frac{M}{N}}g(z)dz+O(\frac{1}{N}).
\end{gather*}
Using

\begin{gather*}
\ln f_{N}\left(y\right)\sim\ln f\left(y\right)-\frac{1}{N}\mathrm{\mathit{\frac{d}{dy}\ln}M}\left(y\right)
\end{gather*}
we have
\begin{gather*}
\sum_{m=0}^{_{^{|A|-1}}}\ln f_{N}\left(\frac{m}{N}\right)\sim N\int_{0}^{y}\ln f(z)dz+\frac{\ln\Lambda(0)\mathrm{M}(0)-\ln\Lambda(y)\mathrm{M}(y)}{2},
\end{gather*}
where $y=\frac{|A|}{N}$ .

So

\begin{gather*}
P_{N}(A)\sim\frac{1}{Z_{N}\sqrt{\Lambda(y)\mathrm{M}(y)}}\exp\left\{ N\int_{0}^{y}\ln f(z)dz\right\} .\;\;\;\;(A.3)
\end{gather*}

Using Stirling's formula $Z_{N}$ we have for the partition function 

\begin{gather*}
Z_{N}\sim\int_{0}^{1}\sqrt{\frac{1}{2\pi N(1-y)y\Lambda(y)\mathrm{M}(y)}}\exp\left\{ NF(y)\right\} dy.\;\;\;\;(A.4)
\end{gather*}

For the probabilities $P_{A,N}(B)=\mathbf{\mathrm{Pr}}_{N}\{\xi_{N}(x,t)=1,\, x\in B,\,\xi_{N}(x,t)=0,\, x\in A\setminus B$\},$\,$
$B\subseteq A\subset\Omega_{N},$ we have 

\begin{gather*}
P_{A,N}(B)=\frac{1}{Z_{N}}\,\sum_{B\leq n\leq N-|A|+|B|}\,\frac{(N-|A|)!}{\left[N-|A|-n+|B|\right]!(n-|B|)!}\exp\left\{ \sum_{m=0}^{n-1}\ln f_{N}\left(\frac{m}{N}\right)\right\} .
\end{gather*}
Now 

\begin{gather*}
\frac{(N-|A|)!}{\left[N-|A|-n+|B|\right]!(n-|B|)!}\sim(1-y)^{|A|-|B|}y^{|B_{}|}\frac{N!}{(N-n)!n!},\: y=\frac{n}{N},
\end{gather*}
and again using Stirling's formula we have

\begin{gather*}
P_{A,N}(B)\sim\frac{1}{Z_{N}}\:\int_{\frac{|\mathit{B|}}{N}\leq y\leq1-\frac{|A|-|B|}{N}}\,(1-y)^{|A|-|B|}y^{|B_{}|}\frac{\exp\left\{ NF(y)\right\} }{\sqrt{2\pi\mathit{N}(1-\mathit{y})\mathit{y\mathrm{\Lambda}\mathrm{(}y\mathrm{)}\mathrm{M}\mathrm{(}y\mathrm{)}}}}dy.
\end{gather*}

By using the Laplace's method we have

\begin{gather*}
\int_{\frac{|\mathit{B|}}{N}\leq y\leq\frac{|A|-|B|}{N}}\,(1-y)^{|A|-|B|}y^{|B|}\frac{\exp\left\{ NF(y)\right\} }{\sqrt{2\pi\mathit{N}(1-\mathit{y})\mathit{y}\Lambda(y)\mathrm{M}(y)}}dy\sim\\
\sim(1-p)^{|A|-|B|}p^{|B|}\frac{\exp\left\{ NF(p)\right\} }{N\sqrt{F^{\prime\prime}(p)(1-\mathit{p})\mathit{p\mathrm{\Lambda}\mathrm{(}p\mathrm{)}\mathrm{M}\mathrm{(}p\mathrm{)}}}}
\end{gather*}
and 

\[
Z_{N}\sim\frac{\exp\left\{ NF(p)\right\} }{N\sqrt{F^{\prime\prime}(p)(1-\mathit{p})\mathit{p\mathrm{\Lambda}\mathrm{(}p\mathrm{)}\mathrm{M}\mathrm{(}p\mathrm{)}}}}.
\]
Therefore

\[
P_{A,N}(B)\sim(1-p)^{|A|-|B|}p^{|B|}
\]
proves the asymptotic independence.

From

\begin{gather*}
F^{\prime}(p)=\ln f(p)-\ln\frac{p}{1-p}=0
\end{gather*}
it follows that

\[
p=\frac{f(p)}{1+f(p)}.
\]

\medskip{}
 $Thanks$: Authors thank the FAPESP for support (grant 2009/15886-9
and 2009/15942-6),

\noindent \begin{flushleft}
Authors are very grateful to V.A.Malyshev and S.A.Pirogov for their
attention to this work.
\par\end{flushleft}

\begin{flushleft}
. 
\par\end{flushleft}
\end{document}